\documentclass[runningheads]{llncs}

\newcommand{\remove}[1]{}
\usepackage[normalem]{ulem}

\usepackage[english]{babel}

\usepackage{tikz}

\usepackage{amsmath}
\usepackage{amssymb}
\usepackage[colorlinks=true, allcolors=blue]{hyperref}

\usepackage{algorithm}
\usepackage{algpseudocode}
\usepackage{cleveref}
\usepackage{comment}
\usepackage{float}
\usepackage{ulem}
\usepackage{color}
\definecolor{Red}{rgb}{1,0,0}
\definecolor{Blue}{rgb}{0,0,1}
\definecolor{Purple}{rgb}{0.7,0.0,1.0}
\definecolor{Green}{rgb}{0,0.5,0}
\definecolor{myGreen}{rgb}{0,0.8,0.4}

\newcommand{\G}{\tilde G}
\newcommand{\UDST}{UVDST}
\newcommand{\SSPT}{SSPT}

\begin{document}
\title{
The Steiner Shortest Path Tree Problem\thanks{Partially supported by the Rita Altura trust chair in computer science, Frankel center for computer science, BGU-NJIT Inst. for Future Technologies, the Israeli Smart Transportation Research Center (ISTRC), and Israeli Science Foundation (Grant No. 465/22).}} 

\author{
Omer Asher\inst{1}
\and Yefim Dinitz\inst{1} 
\and Shlomi Dolev\inst{1}
\and\\ Li-on Raviv\inst{2}
\and Baruch Schieber\inst{3}
}
\titlerunning{Shallowest Shortest Path Tree}
\authorrunning{Asher, Dinitz, Dolev, Raviv, and Schieber}
\institute{Dept. of CS, Ben-Gurion Univ. of the Negev,
Beer-Sheva, Israel\\ 
\and Gilat, Israel\\
\and Dept. of CS, New Jersey Inst. of Technology, Newark, NJ, USA
}
\maketitle

\begin{abstract}
We introduce and study a novel problem of computing a shortest path tree with a minimum number of non-terminals. It can be viewed as an (unweighted) \textit{Steiner Shortest Path Tree} (\SSPT) that spans a given set of terminal vertices by shortest paths from a given source while minimizing the number of nonterminal vertices included in the tree. 
This problem is motivated by applications where shortest-path connections from a source are essential, and where reducing the number of intermediate vertices helps limit cost, complexity, or overhead. 

We show that the \SSPT\ problem is NP-hard. To approximate it,
we introduce and study the \textit{shortest path subgraph} of a graph.
Using it, we show 
an approximation-preserving reduction of \SSPT\ to the uniform vertex-weighted variant of the Directed Steiner Tree (DST) problem, termed \UDST. 
Consequently, the algorithm of [Grandoni et al., 2023] approximating DST implies a quasi-polynomial polylog-approximation algorithm for \SSPT. 
We present a \textit{polynomial} polylog-approximation algorithm for \UDST, and thus for \SSPT\ 
for a restricted class of graphs.
\end{abstract}

\section{Introduction}
Given an undirected or directed graph with nonnegative edge weights and a source vertex $s$, a shortest path tree from $s$ can be found by classic algorithms. 
We introduce the problem of finding a shortest path tree that spans a given set of vertices, called terminals, while minimizing the number of nonterminal vertices in the tree. 
This problem is closely related to the classical unweighted Steiner Tree problem, but differs in that it designates a root vertex and enforces that all paths from the root to each terminal are the shortest.
Due to this relation, we term the problem the Steiner Shortest Path Tree problem (\SSPT). 
This problem is motivated by scenarios where shortest-path connections from a source are essential and where reducing the number of intermediate vertices helps limit cost, complexity, or overhead. 

\subsection{Motivation}
\SSPT\ is a fundamental graph problem that has some obvious applications. One such application is multicasting from a source (root) to a given set of terminal nodes, while simultaneously minimizing latency by restricting paths to the shortest and reducing the total number of bits transmitted by limiting the use of nonterminal nodes. 
In the context of secure communication, it is desirable to minimize the number of nonterminal nodes exposed to the multicast message while still achieving the fastest possible delivery.

An additional optimization application is in \textit{hierarchical supply chain management}~\cite{miller2002}. Consider an airplane manufacturer with service centers located in select airports (terminals). 
To minimize the ground time of airplanes, these service centers must ensure the timely availability of repair parts. To replenish the repair parts in these service centers, the manufacturer may use \textit{hierarchical hubs}. When a part is needed in a service center, it checks its availability in the nearest hub in the hierarchy, which may be used to replenish parts in several service centers. Designing the most cost-effective hierarchical supply chain from a given set of potential hubs gives rise to the \SSPT\ problem.  

\subsection{Our Results and Solution Approach}
We prove that \SSPT\ is NP-Hard. To obtain an approximation, we first introduce and study the concept of the \textit{subgraph of shortest paths} of a weighted graph, which is the union of all the shortest paths from the source. 
Using this concept, we show an approximation-preserving reduction of \SSPT\ to the uniform vertex-weighted variant of the Directed Steiner Tree (DST) problem \cite{charikar1999approximation}, termed \UDST. By this reduction, any approximation for \UDST\ implies a similar approximation for \SSPT.

In particular, the DST algorithm of \cite{charikar1999approximation} achieves an $O(k^\varepsilon)$-approximation in polynomial time for any fixed $\varepsilon > 0$, where $k$ is the number of terminals. The DST algorithm of~\cite{GLS2023} 
achieves an \( O(\log^2 k/\log\log k) \)-approximation in quasi-polynomial time.
Thus, the same results hold for \SSPT.

We suggest a \textit{polynomial} polylog-approximation algorithm solving \UDST\ on polylogarithmically shallow graphs. For logarithmically shallow graphs,
our algorithm achieves an $O(\log^2 k)$-approximation ratio. We note that random graphs, expander graphs, and Small World graphs (i.e., the Internet network) are polylogarithmically shallow graphs. We extend these results to \SSPT, in which case we require that there exists a polylogarithmically (or logarithmically) shallow shortest path tree rooted at the source vertex. 
We extend these algorithms to apply to the general vertex-weighted variant of DST, 
as well as to the weighted variant of SSPT, achieving comparable performance as above on specific classes of graphs.

\subsection{Related Work}
\label{s:related}

Our problem shares conceptual connections with various Steiner tree variants.
%
In the {\it shallow-light tree} problem~\cite{elkin2015steiner}, given an undirected graph and a root, the goal is to construct a spanning tree in which the distance from the root to every vertex is at most an $\alpha$-approximation of the original distance in the input graph, while the total weight of the tree is at most $\beta$ times the weight of a minimum spanning tree. Similarly, the {\it shallow-light Steiner tree} problem \cite{elkin2015steiner} generalizes this to a designated subset of terminal vertices, and further allows the distance bounds to be specified by an arbitrary metric, while its total weight must be within a factor of $\beta$ of the optimal Steiner tree. 
While shallow-light variants focus on bounding stretch and minimizing total edge weight, \SSPT\ minimizes the number of nonterminal vertices under strict path constraints, which sets it apart from the shallow-light family.

A related but distinct variant is the {\it hop-constrained Steiner tree} problem \cite{gouveia2011modeling} that introduces a constraint on the number of edges allowed in the path from the root to each terminal, while still aiming to minimize the total weight of the resulting tree. However, unlike \SSPT, it permits any paths that satisfy the hop constraint, regardless of whether they are the shortest paths. \SSPT\ is more restrictive in this regard, as it permits only the shortest paths and still aims to minimize the number of nonterminal vertices.

\section{\SSPT\ Problem Definition and Hardness}
\label{s:def}

Let us define the Steiner Shortest Path Tree (\SSPT) problem. Its instance $(G,s,X)$ is an undirected or directed nonnegatively weighted graph $G = (V,E,w)$, $w: E \to \mathbb{R}^+$, where a source vertex $s$ and a subset of vertices $X$, $s \notin X$, called terminals, are distinguished.
The goal is to find a shortest path tree w.r.t.\ $w$ from $s$ to the terminals that has the minimum number of nonterminal vertices.
Note that the undirected and directed versions of \SSPT\ are somewhat similar to the undirected and directed Steiner tree problems, but with the specifics of requiring the paths in feasible trees to be shortest.
The weighted variant of \SSPT\ is obtained by assigning weights to the non-terminals and modifying the objective to be finding a shortest path tree that minimizes the overall weight of the non-terminals.

Next, we prove that \SSPT\ is NP-Hard by a reduction from
the classic Set Cover problem. The input to the Set Cover problem consists of a universal set $U$
and a set of its subsets, denoted $\cal S$. The goal is to find a minimum cardinality subset of $\cal S$ that covers $U$. 
Given an instance of the Set Cover problem, the corresponding \SSPT\ instance is constructed as follows. 
Consider the usual representation of $I$ by the bipartite graph with the sides $\cal S$ and $U$ and the edges $(S,x)$ for all $S \in {\cal S}, x \in S$. We add to it the vertex $s$ with the edges $(s,S)$ for all $S \in \cal S$.
All edges are of weight 1.
Observe that any tree rooted at $s$ with edges from $s$ to $\cal S$ and from $\cal S$ to $U$ only
is a shortest path tree in $G$, and vice versa. 
Thus, the reduction provides a natural one-to-one correspondence between the feasible solutions of the instances of Set Cover and \SSPT, while maintaining the objective function value. 
Clearly, this correspondence keeps the optimality of solutions, as required. Note that the approximation ratios of feasible solutions are kept as well.

\section{Shortest Path Subgraph} 
\label{s:SPS}
Consider a weighted, undirected or directed graph $G=(V,E,w)$ 
with a vertex $s$ distinguished as its source. Many problems related to the shortest paths from $s$ to other vertices were studied over several decades. 
A reasonable approach to pruning the input of some of these problems is to introduce the \textit{shortest path subgraph} rooted at $s$. We suggest and study such an object in this section.

The layered network data structure of the \textit{unweighted} shortest paths from $s$ to a target vertex $t$ was introduced by Dinitz in \cite{Dinic70} 
and used as the basis of Dinitz's network flow algorithm \cite{Dinic70} and many following network flow algorithms. 
Its generalization to the layered network data structure of the shortest paths from $s$ to all other vertices had been briefly described at the end of the original paper \cite{Dinic70} and was elaborated and widely developed by Goldberg and Tarjan as the basis for the push-relabel network flow algorithm \cite{DBLP:journals/jacm/GoldbergT88}.

The generalization of the layered network to the case of \textit{weighted} graphs was suggested by Dinitz et al.\ in \cite{DBLP:journals/corr/abs-2101-11514} for the shortest paths from $s$ to a fixed vertex $t$. 
Here, we generalize the approach of \cite{DBLP:journals/corr/abs-2101-11514} to the shortest paths from $s$ to all other vertices of $G$. We consider $G$ to be a directed graph (digraph). Note that this also covers the case of undirected graphs, by the usual reduction that replaces every undirected edge by the pair of oppositely directed edges between the same vertices.

For any pair of vertices $x$ and $y$, we denote by $d(x,y)$ the \textit{distance}, that is, the length of the shortest path from $x$ to $y$ in $G$. We also use the abbreviation $d(v) = d(s,v)$.
We define the 
\textit{shortest path subgraph} rooted at $s$, $\tilde G(s)$, as the union of all shortest paths from $s$ to all other vertices in $G$.

\begin{theorem}
	\label{t:SPS}
\begin{enumerate}
	\item The subgraph $\G(s)$ consists of all vertices reachable from $s$ in $G$ and all edges $(u,v) \in E$ satisfying $d(u) + w(u,v) = d(v)$.
		\item
        Any path from vertex $x$ to vertex $y$ in $\tilde G(s)$ has weight $d(y) -d(x)$ and is a shortest path from $x$ to $y$ in $G$.
\end{enumerate}
\end{theorem}

\begin{proof} 
(1) The substatement on the vertices of $\G(s)$ is straightforward. Any edge $(u,v)$ in a shortest path satisfies $d(u) + w(u,v) = d(v)$ due to the known property of any prefix of a shortest path from $s$ being a shortest path to its end-vertex.
Suppose that $(u,v)$ satisfies $d(u) + w(u,v) = d(v)$. Let $P$ be a shortest path from $s$ to $u$. Appending $(u,v)$ to $P$, we obtain a path from $s$ to $v$ of length $d(u) + w(u,v) = d(v)$, which is a shortest one. Hence, the edge $(u,v)$ is in $\G(s)$.

(2) Consider a path $P$ in $\tilde G(s)$ from $x$ to $y$. Let $P = (x=v_0, v_1, v_2, \dots, v_k=y)$. Then, $w(P) = w(v_0, v_1) + w(v_1, v_2) + \dots + w(v_{k-1}, v_k) = (d(v_1) - d(v_0)) + (d(v_2) - d(v_1)) + \dots + d(v_k) - d(v_{k-1}) = d(v_k) - d(v_0) = d(y) -d(x)$, as required. 

Assume, to contrary, that there is a path $P': w(P') < w(P) = d(y) -d(x)$, from $x$ to $y$ in $G$. Since $x$ is reachable from $s$, there exists path $P_x: w(P_x) = d(x)$, in $G$. Then, the concatenation of $P_x$ and $P'$ is a path from $s$ to $y$ of weight $w(P_x) + w(P') < d(x) + (d(y) - d(x)) = d(y)$ in $G$, a contradiction to the definition of $d(y)$. \qed
\end{proof}

Let us describe some properties of the shortest path subgraph $\tilde G(s)$.
By Theorem~\ref{t:SPS}, $\G(s)$ can be constructed by the execution of algorithm Dijkstra and a simple scan of $G$; thus, the construction time is dominated by the running time of algorithm Dijkstra, which is known to be almost linear in the size of $G$ (see a recent algorithm in \cite{DMMSY25}).

Secondly, by Theorem~\ref{t:SPS}(1), the condition on a path $P$ from $s$ to be shortest in $G$ is \textit{equivalent} to the condition that $P$ is in $\G(s)$.
This relation can simplify the shortest path-related problems. In particular, any instance $(G,s,X)$ of the \SSPT\ problem is equivalent to finding, in the \textit{unweighted} digraph $\G(s)$, a tree rooted at $s$ and spanning $X$ with a minimum number of nonterminal vertices in it.
This equivalence gives rise to the reduction of SSPT that we present in Section~\ref{ss:reductionDST}.



We can use a visual layout of $G$ where each vertex $v$ is placed at a plane point with the horizontal coordinate $d(v)$.
Note that if all edge weights are 1, then the graph $\G$ is \textit{layered}: the vertices are arranged in layers by the integer horizontal coordinates and every edge goes from one layer to the next. 
If the edge weights are positive, then every edge goes from left to right; consequently, the graph $\G$ is \textit{acyclic}. It is also acyclic in the case where the edge weights are nonnegative and there is no cycle consisting of edges with weights zero only in $G$. 
In the general case, there may be nontrivial strongly-connected components of $\G(s)$ containing edges of weight zero only.

We stress that the shortest path subgraph $\G(s)$ is a \textit{digraph} even if $G$ is undirected. To see why, consider the following example. Let $G$ be an undirected 4-cycle $(s,v_1, v_2, v_3)$ with all edge weights 1. Note that the two undirected paths $(s,v_1, v_2)$ and $(s,v_3, v_2)$ are shortest. Thus, $\G(s)$ is composed of these two directed paths. If we remove the edge directions, we could reach $v_3$ via the path $(s,v_1, v_2, v_3)$ in the updated $\G(s)$, but this path is not shortest, contradicting Theorem~\ref{t:SPS}(2). 

We can refine the pruning of $G$ to $\G(s)$ for some problems related to shortest paths. In the case where only the shortest paths from $s$ to the vertices in a given subset $X \subset V$ are of interest in a given problem, the further pruning of $\G(s)$ to the union of all the shortest paths from $s$ to the vertices in $X$, denoted $\G(s,X)$, is applicable.
Note that, given $\G(s)$, the subgraph $\G(s,X) \subseteq \G(s)$ can be constructed in linear time by running a DFS scan from the vertices in $X$ in $\G(s)$ with the reversed direction of all edges, followed by removing from $\G(s)$ all the vertices not reached by this scan.

\section{Approximations of the \SSPT\ Problem}
\label{s:as}

We first show that the \SSPT\ can be reduced to the uniform vertex-weighted version of the directed Steiner tree problem (\UDST), and that this reduction preserves approximability. Thus, we can leverage the known approximation algorithms for the directed Steiner tree problem. Next, we consider a special class of graphs that we call polylogarithmically shallow graphs and show a polylog-approximation algorithm for \UDST\ on this class. This also implies a polylog-approximation algorithm for \SSPT\ on the class of polylogarithmically shortest path shallow graphs.

\subsection{Reduction of \SSPT\ to \UDST}
\label{ss:reductionDST}

The input to the directed Steiner tree problem (DST), studied in \cite{charikar1999approximation,GLS2023},
is a directed graph whose edges have non-negative weights, a source vertex $s$ and a subset of vertices $X$. The goal is to find  a tree with the minimum total edge weight rooted at $s$ and spanning $X$.
We can assume, w.l.o.g., that all vertices are reachable from $s$ (the unreachable ones may safely be removed).

We define two related problems on \textit{vertex-weighted} graphs $G=(V,E,W)$, $W: V \rightarrow \mathbb{R}^+$, $W(t)=0$ for all $t \in X$, which are equivalent to 
variants of DST.
In the general variant, termed VDST, the weights are assigned to the nonterminal vertices arbitrarily, and the goal is to find a tree with the minimum total vertex weight rooted at $s$ and spanning $X$. 
Note that VDST is equivalent to the particular case of DST where the weights assigned to edges $(u,v)$ are the same for all edges incoming to $v$, for any $v \in V$. 
The uniform case of VDST (UVDST) is given by fixing the weight of all non-terminals to 1.
Thus, UVDST is the \textit{minimum non-terminals directed Steiner tree problem} defined as follows: given a digraph, a source vertex $s$, and a set of terminals $X$, $s \notin X$, find a tree rooted at $s$ and spanning $X$ with the minimum number of non-terminal vertices.

By the analysis in Section~\ref{s:SPS}, the \SSPT\ problem on $(G,s,X)$ is equivalent to the \UDST\ problem on $(\G(s),s,X)$ (or on $(\G(s,X),s,X)$), keeping the objective function.
Therefore, any $\alpha$-approximation of the latter problem is an $\alpha$-approximation of the former problem. We thus arrive at the following reduction yielding a family of approximation algorithms for the \SSPT\ problem.

\begin{theorem}
\label{t:red-DST}
Let $\cal A$ be any approximation algorithm solving the \UDST\ problem with approximation ratio $\alpha$. Then, given an instance $I=(G, s, X)$ of the \SSPT\ problem, running $\cal A$ on the instance $(\G(s,X),s,X)$ (or $(\G(s,X),s,X)$) of \UDST\ returns a tree which is an $\alpha$-approximation to the optimum for $I$, keeping the same running time bound.
\end{theorem}

In particular, the quasi-polynomial $O(\log^2(|X|)/\log \log |X|)$-approximation algorithm in~\cite{GLS2023} yields
a quasi-polynomial $O(\log^2(|X|)/\log \log |X|)$-approximat\-ion algorithm for the \SSPT\ problem.

\noindent

\textit{Remark\/}:
To further compare the complexity of the problems discussed, we review some additional observations on (approximation-preserving) reductions among undirected SSPT (U-SSPT), directed SSPT (D-SSPT), and \UDST.
The general \UDST\ can be reduced to D-SSPT by setting all edge weights to be 0.
We are aware of a reduction from \UDST\ to U-SSPT only when the input graph is acyclic. The reduction sets the weight of every edge $(u,v)$ in the U-SSPT instance to be $D(s,v) - D(s,u)$, where $D(s,x)$ is the maximum length of a path from $s$ to $x$. 
As for the comparison between U-SSPT and D-SSPT, the reduction from U-SSPT to D-SSPT is straightforward, by replacing every edge by the two oppositely directed edges between the same vertices.
D-SSPT can be reduced to U-SSPT if there are no zero cycles (cycles with edges of weight zero only) in the graph, similarly to the reduction from \UDST\ to U-SSPT as above.

\subsection{Polynomial  Polylog Approximations for Polylogarithmically Shallow Graphs}
\label{ss:polylogshallow}

In this section, we suggest a polynomial polylog-approximation algorithm for the \UDST\ problem on polylogarithmically shallow graphs (to be defined precisely below). This implies, by Theorem~\ref{t:red-DST}, a similar result for the \SSPT\ problem.
By saying that $T$ is a tree, we mean a tree rooted at $s$ whose edges are directed along the paths from $s$. 
For any tree $T$, let $nt(T)$ denote the number of non-terminals in $T$.
We denote by $OPT(I)$ the number of non-terminals in the optimal solution to any problem instance $I$. 

Our plan of approximating \UDST\ is as follows. Given an instance $I=(G,s,X)$ of \UDST, we first construct an instance $I^{\mbox{\scriptsize SC}}$ of Set Cover whose optimum $OPT(I^{\mbox{\scriptsize SC}})$ lower bounds $OPT(I)$. Based on the solution to $I^{\mbox{\scriptsize SC}}$, we generate a feasible solution of $I$ whose objective function value is upper bounded by a function of $OPT(I^{\mbox{\scriptsize SC}})$ and, thus, also by the same function of $OPT(I)$.

Let $G[X]$ be the subgraph of $G$ induced by $X$ (that is, composed of the terminals and the edges between them). 
%
%
%
\remove{
This generalizes straightforwardly to the following statement:

\begin{lemma}
\label{l:reach}
    For any tree $T^{\mbox{\scriptsize X}}$ in $G[X]$ rooted at a terminal $t$ and a tree $T$ in $G$ spanning $t$, there exists a tree $T'$ in $G$ spanning all terminals in $T^{\scriptsize \mbox{X}}$ with $nt(T')=nt(T)$.
\end{lemma}
}
Consider the set of strongly-connected components of $G[X]$. (Note that some of these components may be a single terminal.) Define its subset $\cal S$ of \textit{source components} to be those not reachable in $G[X]$ from any other component. 

\begin{lemma}
\label{l:SCC}
   There exists a linear time algorithm that, given any tree $T$ spanning at least one terminal in each component in $\cal S$, builds a tree $T'$ spanning the entire $X$ with $nt(T') = nt(T)$.
\end{lemma}

\begin{proof}
    Let us choose one of the terminals spanned by $T$ in every source component. We construct a DFS forest $F$ in $G[X]$ from all these terminals. 
    By construction, $F$  spans $X$. We build a tree $T'$ as required by adding to $T$ all edges in $F$ leading to the terminals not in $T$. 
    Note that such edge additions do not change the set of non-terminals in the tree.
\qed
\end{proof}

    \remove{
\begin{lemma}
\label{l:bar G}
\begin{enumerate}
    \item There exists a linear time algorithm that, given any tree $\bar T$ spanning $\bar X$ in $\bar G$, builds a tree $T'$ spanning $X$ in $G$ with $nt(T') \le nt(\bar T)$. 
    \item $OPT(G,s,X) = OPT(\bar G,s,\bar X)$.
\end{enumerate}
\end{lemma}

\begin{proof}
    (Item 1) 
    Restore $G$ from $\bar G$ by expanding each contracted component. Thus, the set of edges of $\bar T$ forms a forest $F$ in $G$. 
    Note that exactly one of the terminals in each expanded component has an incoming edge in $F$. We construct a DFS forest $F'$ in $G[X]$ from all the terminals with incoming edges. 
    By construction, $F'$  spans $X$. We build a tree $T'$ as required by adding to $F$ all edges in $F'$ leading to the terminals without incoming edges in $F$. 
    Note that neither the expansion nor the edge addition as above changes the set of non-terminals in the tree,
    while some non-terminals in $\bar T$ may become terminals in $T'$.
    
    (Item 2)
    By item 1, $OPT(G,s,X) \le OPT(\bar G,s,\bar X)$.
    To prove the inverse inequality, consider any tree $T$ spanning $X$ in $G$ and contract in $T$ every source component of $G[X]$. The result may not be a tree. To transform it back to a tree, remove all but one of the edges entering each vertex. Additionally, remove all subtrees that do not contain any terminals in $\bar X$. The resulting tree $\bar T$ spans $\bar X$ in $\bar G$ and contains no more non-terminals than those contained in $T$.  
    \qed
\end{proof}

As a consequence, any $f(N)$-approximation \UDST\ algorithm for graphs without edges between terminal pairs implies an $f(N)$-approximation \UDST\ algorithm for general graphs, where $f(N)$ is any monotone function of the number of terminals in the graph.
}

We use the following observation in designing our approximation algorithm. 
We show that in any tree $T$ feasible for $I$, that is, spanning $X$, the set of edges from non-terminals to terminals ``covers'' $\cal S$, in a sense.
We denote by $Pre\cal S$ the set of non-terminals $\{v \in V \setminus X: \exists (v,t) \in E, t\in S, S\in \cal S\}$
(that is, the vertices in $Pre\cal S$ are the non-terminals in $G$ with outgoing edges to the terminals in the source components of $G[X]$) .
For any $v \in Pre\cal S$, let $N(v) = \{ S \in {\cal S}:  \exists (v,t) \in E, t \in S\}$.
We denote by $I^{\mbox{\scriptsize SC}}$ the instance of Set Cover defined by the subsets $\{N(v): v \in Pre\cal S\}$ of ${\cal S}$.

\begin{lemma}
\label{l:non-terminals}
   For any instance $I$ of \UDST, $OPT(I^{\mbox{\scriptsize SC}})$ is a lower bound on $OPT(I)$.
\end{lemma}

\begin{proof}
Consider any tree $T$ feasible for $I$. Denote by $V'(T)$ the set of non-terminals in $T$ that are also in $Pre\cal S$. 
Let us show that for any source component $S$, there exists an edge in $T$ from $V'(T)$ to a terminal in $S$. Let $t$ be any terminal from $S$ closest to $s$ in $T$. Its parent $v$ in $T$ cannot be in $S$ by our choice and cannot be in another strongly-connected component of $G[X]$ since $S$ is a source component. Hence, $v$ is a non-terminal, and $(v,t)$ is an edge as required. 

Let tree $T^*$ be optimal for $I$. By the above, $\{N(v): v \in V'(T^*)\}$ is a feasible solution for $I^{\mbox{\scriptsize SC}}$. Now, $OPT(I^{\mbox{\scriptsize SC}}) \le |V'(T^*)| \le nt(T^*) = OPT(I)$, as required.
\qed
\end{proof}

\remove{
For every source component $S\in \cal S$, define $t(S)$ to be a terminal from $S$ closest to $s$ in $T$ (we break ties arbitrarily). By construction, the parent of $t(S)$ in $T$ is a non-terminal; let $V'(T)$ be the set of the parents of $t(S)$, over all $S \in \cal S$. In this sense, $V'(T)$ ``covers'' $\cal S$ by the edges outgoing from the vertices in $V'(T)$ in $T$ to the vertices in the subsets in $\cal S$. 
Since $nt(T)$ is at least $|V'(T)|$, the minimum possible size of a cover as above is a lower bound on $OPT(I)$. 

Let us formalize this observation. We denote by $Pre\cal S$ the set of non-terminals $\{v \in V \setminus X: \exists (v,t) \in E, t \in S \in \cal S\}$ 
(that is, the vertices in $Pre\cal S$ are the non-terminals in $G$ with outgoing edges to the terminals in the source components of $G[X]$) .
For any $v \in Pre\cal S$, let $N(v) = \{ S \in {\cal S}:  \exists (v,t) \in E, t \in S\}$.
We denote by $I^{\mbox{\scriptsize SC}}$ the instance of Set Cover defined by the subsets $\{N(v) \subseteq {\cal S}: v \in Pre\cal S\}$.
Let tree $T^*$ be optimal for $I$. Since $\{N(v): v \in Pre{\cal S} \cap T^*\}$ is a feasible solution for $I^{\mbox{\scriptsize SC}}$, $OPT(I^{\mbox{\scriptsize SC}}) \le |Pre{\cal S} \cap T^*| \le nt(T^*) = OPT(I)$. That is, $OPT(I^{\mbox{\scriptsize SC}})$ is a lower bound on $OPT(I)$.
}

Consider the following \textit{algorithm} ${\cal A} = {\cal A}(I)$ for the \UDST\ problem:

\begin{enumerate}
    \item Construct $I^{\mbox{\scriptsize SC}}$ based on $I$.
    \item Find an $O(\log |\cal S|)$-approximation  of $I^{\mbox{\scriptsize SC}}$ by the algorithm in \cite{chvatal1979greedy}, denoted $P$. Note that $P$ is defined by  $V_{\mbox{\scriptsize SC}}\subseteq Pre\cal S$ such that $P=\{N(v): v \in V_{\mbox{\scriptsize SC}}\}$.
    \item 
    Construct an edge set $E_{\mbox{\scriptsize SC}}$ by choosing for each $S \in \cal S$, a single edge $(v,t) \in E$, where $t \in S$ and $v \in V_{\mbox{\scriptsize SC}}$.
    \item Build a BFS tree $T_{BFS}$ from $s$ to $V_{\mbox{\scriptsize SC}}$ in $G$.
    \item Build a tree $T$ by adding to $T_{BFS}$ all edges in $E_{\mbox{\scriptsize SC}}$ leading to the terminals not in $T_{BFS}$.
    \item Return the tree $T'$ obtained from $T$ by expanding it as in the proof of Lemma~\ref{l:SCC}.
\end{enumerate}

To analyze the approximation ratio of this algorithm, 
we define the notion of \textit{shallowness}. 
A digraph $G$ with a source vertex $s$ is $R$-shallow if for every $v \in V$, there exists a path from $s$ to $v$ that has no more than $R$ edges. (Note that the minimum possible value of such an $R$ is the radius from $s$ in $G$.)

\remove{define the \textit{hop-length} of a path $P$ to be the number of edges in $P$.
the \textit{hop-distance} of a vertex $v$ from $s$, $hd(v)$, to be the minimum hop-length of a path from $s$ to $v$,
and the \textit{hop-radius} $R(s,V')$ w.r.t.\ any subset $V'\subseteq V$ to be $\max_{v \in V'} hd(v)$.
}

\begin{theorem}
\label{t:log DST}
The algorithm ${\cal A}$ is polynomial and returns an $O(R \cdot \log(|X|))$-approximation to the \UDST\ problem on $R$-shallow digraphs. 
\end{theorem}

\begin{proof}
The algorithm is polynomial since all its steps and auxiliary objects are polynomial in the input size. 
Step 3 of the algorithm is consistent, and the tree $T$ spans all components in $\cal S$ since $P$ is a set cover of $\cal S$.
The tree $T'$ is feasible for $I$ by Lemma~\ref{l:SCC}.
To prove the approximation ratio, we analyze the size of $T'$. 

\remove{
Let a tree $T^*$ be an optimal solution to $I$, that is, $nt(T^*) = OPT(I)$. 
Note that for any source component $S \in \cal S$, the tree $T^*$ contains an edge from a vertex in $Pre\cal S$ to $S$. Let $V^* \subseteq Pre\cal S$ be the set of tails of all such edges in $T^*$; thus $\{N(v): v \in V^*\}$ is a feasible solution to $I^{\mbox{\scriptsize SC}}$. 
Hence, $|V_{\mbox{\scriptsize SC}}| = O(\log(|{\cal S|}) \cdot OPT(I^{\mbox{\scriptsize SC}}) 
= O(\log(|{\cal S}|) \cdot OPT(I)$, by Lemma~\ref{l:non-terminals}.
}

Note that $T_{BFS}$ is a union of paths with the minimum number of edges from $s$ to each of the vertices in $V_{\mbox{\scriptsize SC}}$; therefore, $nt(T_{BFS})$ is at most $R \cdot|V_{\mbox{\scriptsize SC}}|$.
By Lemmas~\ref{l:SCC} and ~\ref{l:non-terminals}, the objective function value of the returned solution is $nt(T') = nt(T) = nt(T_{BFS}) \le R \cdot |V_{\mbox{\scriptsize SC}}| =  R \cdot O(\log(|{\cal S|}) \cdot OPT(I^{\mbox{\scriptsize SC}}) 
= O(R \cdot \log(|X|) \cdot OPT(I)$, as required.    
\qed
\end{proof}


We say that a directed graph $G$ with a source vertex $s$ is \textit{polylogarithmically shallow} if it is $R$-shallow where $R$ is polylogarithmic in $|X|$. 
\textit{Logarithmically shallow} digraphs are defined similarly. 

\begin{corollary}
\label{c:\UDST}
Algorithm ${\cal A}$ computes a polylog approximation of \UDST\ on polylogarithmically shallow digraphs, and an $O(\log^2 |X|)$-approximation on logarithmically shallow digraphs.
\end{corollary}

For our \SSPT\ result, we define \textit{shortest path shallowness}. A  graph $G$ (either undirected or directed) with a source vertex $s$ is $R$-shortest path shallow if for every $v\in V$ there exists a \textit{shortest} path from $s$ to $v$ with no more than $R$ edges. 
(Note that the minimum possible value of such an $R$ is the radius from $s$ in $\tilde G(s)$.)
We say that $G$ with a source vertex $s$ is \textit{polylogarithmically shortest path shallow} if it is $R$-shortest path shallow where $R$ is polylogarithmic in $|V|$. 
\textit{Logarithmically shortest path shallow} graphs are defined similarly.

\begin{corollary}
\label{c:\SSPT} 
There exists a polynomial algorithm for the \SSPT\ problem on (undirected or directed) instances $(G,s,X)$ that provides a polylog-approx\-imation when $G$ is polylogarithmically shortest path shallow,
and an $O(\log^2 |X|)$-approxim\-ation when $G$ is a logarithmically shortest path shallow.
\end{corollary}

The following generalization of Corollary~\ref{c:\UDST} to VDST and the according generalization of  Corollary~\ref{c:\SSPT} to the weighted SSPT may be of some interest.
To analyze VDST, define the weight $W(P)$ of a path $P$ to be the total weight of the vertices in it, and the distance $D(v)$ to be the minimum weight of a path from $s$ to $v$, $v \in V$.
We modify algorithm ${\cal A}$ to ${\cal A}^W$ by replacing the BFS tree in step 4 by the shortest path tree from $s$ to $V_{SC}$ w.r.t.\ $W$.
Note that the algorithm of \cite{chvatal1979greedy} provides an $O(\log |\cal S|)$-approximation solution also to the weighted Set Cover problem, so algorithm ${\cal A}^W$ is consistent with the VDST problem. 
Extending straightforwardly the analysis in the proof of Theorem~\ref{t:log DST} to this problem, we generalize Corollary~\ref{c:\UDST} to:

\begin{corollary}
There exists a polynomial algorithm for the VDST problem on the instances $(G=(V,E,W),s,X)$ that provides a polylog-approximation when the ratio $D(v)/W(v)$ is polylogarithmic in $|X|$, and an $O(\log^2 |X|)$-approximation when the ratio $D(v)/W(v)$ is logarithmic in $|X|$, for all $v \in V \setminus X$.
\end{corollary}

Recall that in the weighted variant of \SSPT, the non-terminals are 
assigned weights $W$ (which differ from the weight $w$ on the edges) and the goal is to minimize the total weight of the non-terminals in a shortest path tree.

\begin{corollary}
There exists a polynomial algorithm for the weighted (undirected or directed) variant of the \SSPT\ problem 
on the instances $(G=(V,E,w,W),s,X)$ that provides a polylog-approximation when, 
for any vertex $v \in V$, there exists a shortest w.r.t.\ $w$ path $P$ from $s$ to $v$ with the ratio $W(P)/W(v)$ polylogarithmic in $|X|$, and 
an $O(\log^2 |X|)$-approximation when, 
for any vertex $v \in V$, there exists a shortest w.r.t.\ $w$ path $P$ from $s$ to $v$ with the ratio $W(P)/W(v)$ logarithmic in $|X|$.
\end{corollary}

\bibliographystyle{plain}
\bibliography{DS}

\end{document}